\documentclass[12pt, onecolumn]{IEEEtran}
\usepackage{amsfonts}
\ifCLASSINFOpdf \else \fi

\usepackage[utf8]{inputenc} 
\usepackage[dvips]{graphicx}
\usepackage{algorithm}
\usepackage{algorithmic}
\usepackage{amsthm}
\usepackage{stfloats}
\usepackage{cite}
\usepackage{makecell}
\usepackage{multirow}
\usepackage[numbers,sort&compress]{natbib}
\usepackage{ragged2e}
\usepackage{color}
\usepackage{tcolorbox}
\usepackage[justification=centering]{caption}

\newtheorem{Theo}{Theorem}

\IEEEoverridecommandlockouts
  
\begin{document}

\title{5G PRS-Based Sensing: A Sensing Reference Signal Approach for Joint Sensing and Communication System}

\author{Zhiqing Wei, Yuan Wang, Liang Ma, Shaoshi Yang, Zhiyong Feng, \\
Chengkang Pan, Qixun Zhang, Yajuan Wang,
 Huici Wu, Ping Zhang

\thanks{Zhiqing Wei, Yuan Wang, Shaoshi Yang, Zhiyong Feng, Qixun Zhang, Huici Wu and Ping Zhang are with Beijing University of Posts and Telecommunications, Beijing, China 100876 (email: \{weizhiqing, wangyuan, shaoshi.yang, fengzy, zhangqixun, dailywu, pzhang\}@bupt.edu.cn) and Liang Ma, Chengkang Pan, Yajuan Wang are with Future Research Lab, China Mobile Research Institute, Beijing, China (email: \{maliangyjy, panchengkang, wangyajuan\}@chinamobile.com). \emph{Correspondence author: Zhiqing Wei, Zhiyong Feng.}}}
\maketitle

\begin{abstract}

The emerging joint sensing and communication (JSC) technology is expected to support new applications and services, such as autonomous driving and extended reality (XR), in the future wireless communication systems. Pilot (or reference) signals in wireless communications usually have good passive detection performance, strong anti-noise capability and good auto-correlation characteristics, hence they bear the potential for applying in radar sensing. In this paper, we investigate how to apply the positioning reference signal (PRS) of the 5th generation (5G) mobile communications in radar sensing. This approach has the unique benefit of compatibility with the most advanced mobile communication system available so far. Thus, the PRS can be regarded as a sensing reference signal to simultaneously realize the functions of radar sensing, communication and positioning in a convenient manner. Firstly, we propose a PRS based radar sensing scheme and analyze its range and velocity estimation performance, based on which we propose a method that improves the accuracy of velocity estimation by using multiple frames. Furthermore, the Cr\'{a}mer-Rao lower bound (CRLB) of the range and velocity estimation for PRS based radar sensing and the CRLB of the range estimation for PRS based positioning are derived. Our analysis and simulation results demonstrate the feasibility and superiority of PRS over other pilot signals in radar sensing. Finally, some suggestions for the future 5G-Advanced and 6th generation (6G) frame structure design containing the sensing reference signal are derived based on our study. \\
\end{abstract}
\IEEEpeerreviewmaketitle

\IEEEpeerreviewmaketitle

\section{Introduction}
% 第一段是大背景，引出一体化
The 5th generation advanced (5G-A) and 6th generation (6G) mobile communication systems are expected to support novel services such as autonomous driving, extended reality (XR), and so forth \cite{TR:CR0}, which will require powerful communication and sensing capabilities simultaneously. Wireless sensing, including positioning, velocity detection and imaging, has long been an independent technology developed in parallel with mobile communications. In 5G-A and 6G mobile communication systems, millimeter wave (mmWave), terahertz (THz) and massive multi-input multi-output (MIMO) technologies could be indispensable. As a result, the frequency bands and antennas of wireless communication systems are becoming similar to those of radar, which makes the joint sensing and communication (JSC) technology feasible and promising \cite{TR:CR1}. In JSC, sensing and communication functions will be mutually beneficial in the same system, which can improve the spectral and energy efficiency while reducing the hardware cost. The application of JSC technology in future mobile networks has already become a consensus \cite{TR:CR}. For example, the ITU IMT--2030 \cite{TR:CR2} has identified JSC as one of the candidate enabling technologies of 6G.

% 第二段是小背景，引出基于导频方式的一体化波形
Waveform design is fundamental to the JSC technology. It imposes a strong impact on the performance of sensing and communication. The existing schemes can be divided into sensing centric waveforms and communication centric waveforms \cite{TR:CR3}. The sensing centric designs are based on the waveforms of radar, such as the linear frequency modulated (LFM) signal \cite{TR:CR4} and the phase coded signal \cite{TR:CR7,TR:CR8}. However, the transmission rate of them is low. By contrast, the communication centric designs are based on the waveforms of wireless communications. So far, the orthogonal frequency division multiplexing (OFDM) has been widely applied in the design of communication centric JSC waveforms. For example, Sturm $et\;al.$ \cite{TR:CR9,TR:CR10} first proposed a low-complexity signal processing algorithm to independently recover the range and Doppler information in an OFDM based radar system. Furthermore, according to the Shannon information theory \cite{TR:CR28}, typically radar sensing needs structured waveforms with strong auto-correlation characteristics, while wireless communication requires random waveforms for maximizing the data rate. The pilot signals in wireless communication systems usually have good passive detection performance, strong anti-noise capability and good auto-correlation characteristics. Therefore, pilot signals with high auto-correlation have a great potential to be used in radar sensing, hense the JSC waveform design based on pilot signals has attracted much attention. There are various pilots designed for different purposes in the 5G standard, such as \textit{synchronization signal (SS)}, \textit{demodulation reference signal (DMRS)}, \textit{channel state information-reference signal (CSI-RS)}, \textit{positioning reference signal (PRS)}, and so on. Pilot based JSC waveforms have the following advantages: 1) In a JSC system, pilots can be used for radar sensing with low hardware cost and short implementation cycle; 2) The communication function and the pilot-based sensing function share resources in time domain or frequency domain, which can minimize the interference between radar and communication signals \cite{TR:CR17}; 3) The radar signal processing only needs to concentrate on the pilots, which reduces the computational complexity \cite{TR:CR13}.

%第三段：说一下别人是怎么做的
Pilot based JSC waveform design was mainly studied from the perspectives of pilot sequence selection \cite{TR:CR11,TR:CR12}, pilot insertion mode \cite{TR:CR13,TR:CR14,TR:CR15}, and time-frequency resource allocation \cite{TR:CR13,TR:CR16,TR:CR17}. In terms of pilot sequence selection, Kumari $et\;al.$ \cite{TR:CR11} applied the preamble signal of the 802.11ad in radar sensing due to its good auto-correlation property. The Zadoff-Chu sequence with good auto-correlation was applied as the comb pilots of channel estimation to strike a balance between the communication error rate and the radar sensing performance in high mobility scenarios \cite{TR:CR12}. In terms of pilot insertion mode, Ozkaptan $et\;al.$ \cite{TR:CR13} inserted the Barker code in a comb pilot, which was used in radar sensing and channel estimation to achieve high range and velocity resolution. The hardware-level verification in the 76-81 GHz frequency band for the design of \cite{TR:CR13} was provided in \cite{TR:CR14}. A JSC waveform using superimposed pilot and orthogonal space-time block code (OSTBC) precoding was proposed in \cite{TR:CR15}, which is able to detect the target with large Doppler frequency shift. In terms of time-frequency resource allocation, \cite{TR:CR13} optimized subcarrier allocation to maximize the performance of radar sensing and communication under the constraints of radar estimation accuracy and effective channel capacity. The pilot based OFDM waveform is capable of realizing short-range, medium-range and long-range radars based on the flexible allocation of pilot subcarriers in \cite{TR:CR16}. A full-duplex communication and radar sensing system was designed in \cite{TR:CR17} using a frequency-division duplex scheme, where the data rate of communication was maximized by optimizing the power allocation of pilot signal.

However, the location of pilot subcarriers in \cite{TR:CR13} are rearranged in a stepped manner, and the superimposed pilot used in \cite{TR:CR15} introduces interference to the data by directly adding the pilot sequence to the data sequence, which does not meet the communication standard. In addition, the pilot in \cite{TR:CR11} is based on IEEE 802.11ad WLAN standard, while the OFDM system used in \cite{TR:CR13,TR:CR14,TR:CR16} is at 79 GHz carrier frequency and MHz large subcarrier frequency spacing to achieve high radar resolution, which makes them incompatible with 5G New Radio (NR) standard. In fact, there is a scarcity of work that directly investigate the use of the modern mobile communication signals of 5G and 5G-A in radar sensing \cite{TR:CR19}. The ITU-R WP5D Document \cite{TR:CR20} indicates that IMT-2030 and beyond will consider integrated communication, positioning and sensing, and potentially even jointly design flexible signals for concurrent communication, positioning and sensing with slight or no modification to hardware and waveform. Meanwhile, the sensing based positioning techniques can be further combined with the traditional reference signal based positioning techniques to improve the accuracy.

%第五段介绍本文的研究工作（列出1、2、3）和优势
In this paper, considering that PRS is a signal specially designed for 5G in-band wireless positioning, we accomplish radar sensing using 5G PRS. Thus, PRS can be regarded as a \textit{sensing reference signal} to simultaneously realize the functions of radar sensing and positioning. PRS enjoys the advantages of long sequence, good auto-correlation, rich time-frequency resources and flexible configuration, which is more suitable for radar sensing compared with other pilot signals. The application of PRS in radar sensing does not require additional modification of communication signals, thus it is well compatible with 5G mobile communication systems, which is beneficial for promoting the integration of JSC technology in the coming 5G-A mobile communication system\footnote{5G-Advanced will start from 3GPP Release 18 (R18), which is expected to be frozen by the end of 2023. Hence, the sensing reference signal based JSC signal design will promote the rapid implementation of JSC technology.}. As shown in Fig. \ref{fig_JRCN}, PRS is sent in the downlink to detect the target. The difficulty of engineering application may be the radar processing at the base station (BS). Due to the compatibility between the PRS based sensing and 5G communication system, there will be few modifications to the hardware. In 2021, Huawei completed the world's first verification of 5G-A JSC technology \cite{TR:CR35}. The test results show that the detection range of JSC base station exceeds 500 meters, which will provide an important reference for the integration of sensing capability at the base station. The main contributions of this paper are as follows.
\begin{figure}[!t]
\centering
\includegraphics[width=0.7\textwidth]{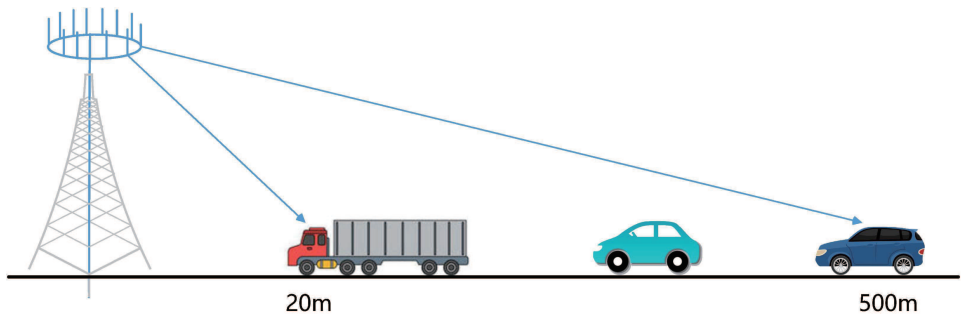}
\caption{JSC based roadside base station detecting vehicles} \label{fig_JRCN}
\end{figure}

1) We consider the downlink reference signals commonly used in 5G NR, namely SS, DMRS and CSI-RS, as comparisons to demonstrate the feasibility of PRS for radar sensing. Then the PRS based OFDM signal is applied as a JSC waveform, which can be regarded as the sensing reference signal. Our radar signal processing is focused on the pilot subcarriers to realize the range and velocity estimation, thus there is no need to make any changes to the 5G signal.

2) The Cr\'{a}mer-Rao lower bound (CRLB) of the range and velocity estimation concerning the PRS based sensing is derived, and the method of increasing the number of Fourier transform points \cite{TR:CR21} is applied to improve the accuracy of range and velocity estimation. In addition, an approach of velocity measurement using multiple frames is proposed to improve the accuracy of velocity estimation and reduce the time slot overhead within a frame.

3) The CRLB of PRS based positioning is derived. Positioning is the main sensing service provided by 5G mobile communication system. The multi-functional integration of sensing, communication, and positioning is realized with the aid of PRS.

4) Based on the above analysis, we give suggestions to the frame structure design in 5G-A and 6G. The sensing reference signal is configured in the frame to perform sensing tasks. The flexible deployment of sensing and communication overhead to satisfy the requirements of different scenarios, and the tradeoff between range and velocity estimation in radar sensing, are also discussed.

The rest of this paper is organized as follows. Section II introduces the PRS of 5G and evaluates its feasibility in radar sensing. Section III describes the system model of PRS based sensing. Section IV presents our proposed radar signal processing method based on PRS. The CRLBs of radar sensing and positioning relying on PRS are derived in Section V. Simulation results and discussions are presented in Section VI. The conclusions are drawn in Section VII. For convenience, the key parameters involved in this paper are listed in Table I.

\begin{table}[!t]
 \caption{\label{sys_para}Key Parameters}
 \begin{center}
 \begin{tabular}{l l}
 \hline
 \hline

    {Symbol} & {Description} \\
 %Symbol & Description & Values\\

  \hline

  $\mu$ & Subcarrier spacing configuration  \\
  $\Delta f$ & Frequency spacing of subcarriers  \\
  ${N_{{\rm{symb}}}^{{\rm{slot}}}}$ & Number of symbols per time slot \\
  $N_{{\rm{slot}}}^{{\rm{frame,}}\mu }$ & Number of slots per frame for $\mu$ \\
  ${T_{{\rm{CP}}}}$ & Length of the cyclic prefix \\
  $T$ & Duration of the OFDM symbol without cyclic prefix \\
  $T_s$ & Total duration of OFDM symbol \\
  $M$ & Number of OFDM symbols\\
  $N$ & Number of subcarriers \\
  $J$ & Set of subcarriers carrying PRS \\
  $N_J$ & Number of subcarriers carrying PRS \\
  $K_{{\rm{comb}}}^{{\rm{PRS}}}$ & Comb size of PRS\\
  $f_c$ & Carrier frequency  \\
  $c$ & Speed of light \\
  $\xi $ & Attenuation factor \\
  $\tau$ & Round-trip delay of target\\
  $R_r$ & Range of target\\
  $f_{d,r}$ & Doppler shift of target \\
  $v$ & Velocity of target\\
  $\Delta R$ & Range resolution \\
  ${R_{\max }}$ & Maximum unambiguous range \\
  $m_a$ & Fractional factor \\
  $\Delta v$ & Velocity resolution \\
  ${v_{\max }}$ & Maximum unambiguous velocity \\
  $SNR$ & Signal-to-noise ratio \\

  \hline
  \hline
 \end{tabular}
 \end{center}
\end{table}
\section{5G Positioning Reference Signal}
In this section, we introduce the PRS of 5G NR, and analyze its sequence correlation and time-frequency resource mapping scheme to evaluate its feasibility in radar sensing.
\subsection{PRS Sequence Generation}
The PRS sequence generating equation defined in TS 38.211 \cite{TR:CR22} is as follows.
\begin{equation}
r(m) = \frac{1}{{\sqrt 2 }}(1 - 2c(2m)) + j\frac{1}{{\sqrt 2 }}(1 - 2c(2m + 1)),
\end{equation}
where $j$ is the imaginary unit, and the pseudo-random sequence $c$($i$) adopts the 31st-order Gold sequence. The generating equation of the initial sequence of $c$($i$) in \cite{TR:CR22} is
\begin{equation}
\begin{array}{*{20}{l}}
{{c_{{\rm{init}}}} = ({2^{22}}\left\lfloor {\frac{{n_{{\rm{ID}},{\rm{seq}}}^{{\rm{PRS}}}}}{{1024}}} \right\rfloor  + {2^{10}}(N_{{\rm{symb}}}^{{\rm{slot}}}n_{s,f}^\mu  + l + 1) \cdot (2(n_{{\rm{ID}},{\rm{seq}}}^{{\rm{PRS}}}}\\
{\,\bmod \,1024) + 1) + (n_{{\rm{ID}},{\rm{seq}}}^{{\rm{PRS}}}\,\bmod \,1024))\,\bmod \,{2^{31}},}
\end{array}
\end{equation}
where ${n_{{\rm{ID}},{\rm{seq}}}^{{\rm{PRS}}}} \in \{ 0,1, \ldots ,4095\}$ is the downlink PRS sequence ID uniquely identifying a downlink PRS resource, ${N_{{\rm{symb}}}^{{\rm{slot}}}}$ is the number of symbols per time slot, $n_{s,f}^\mu  \in \{ 0, \ldots ,N_{{\rm{slot}}}^{{\rm{frame,}}\mu } - 1\}$ is the number of time slots within a frame, assuming the subcarrier spacing configuration $\mu$ and the number of time slots per frame $N_{{\rm{slot}}}^{{\rm{frame,}}\mu }$, and $l$ is the number of OFDM symbols in the time slot to which the sequence is mapped.

Since PRS is generated by the Gold sequence and the Gold sequence has good auto-correlation and cross-correlation characteristics, as shown in Fig. \ref{fig_PRS correlation characteristics}, we infer that PRS is suitable for radar sensing.
\begin{figure}[!t]
\centering
\includegraphics[width=0.6\textwidth]{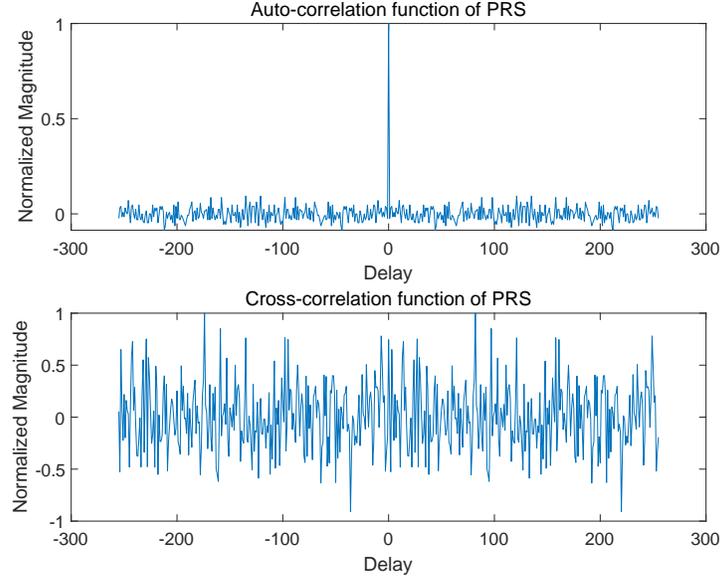}
\caption{PRS correlation characteristics} \label{fig_PRS correlation characteristics}
\end{figure}
\subsection{Time-Frequency Resource Mapping for PRS}
Based on the fact that a physical resource block (PRB) is defined
as 12 continuous subcarriers in frequency domain, the PRBs allocated to PRS in TS 38.214 \cite{TR:CR23} have a granularity of 4 PRBs, a minimum of 24 PRBs and a maximum of 272 PRBs. Table II shows the flexible transmission parameters supported in 5G NR. It is worth mentioning that $\Delta f$ is the frequency spacing of subcarriers, ${T_{{\rm{CP}}}}$ denotes the length of the cyclic prefix (CP), and $T$ denotes the duration of an OFDM symbol without CP. The 3GPP standard divides the frequencies available to 5G into the Frequency Range 1 (FR1, 450 MHz - 5.9 GHz) band and the Frequency Range 2 (FR2, 24.2 GHz - 52.6 GHz) band.
\begin{table}[htbp]
    \renewcommand\arraystretch{1.2}
    \caption{Flexible transmission parameters supported by 5G NR \textsuperscript{\cite{TR:CR22,TR:CR24}}}
    \begin{center}
        \begin{tabular}{|c|c|c|c|c|c|}
        \hline
        $\mu $&0&1&2&3&4\\
        \hline
        $\Delta f = {2^\mu } \cdot 15\left[ {{\rm{kHz}}} \right]$&15&30&60&120&240\\
        \hline
        ${N_{{\rm{symb}}}^{{\rm{slot}}}}$&14&14&14&14&14\\
        \hline
        $N_{{\rm{slot}}}^{{\rm{frame,}}\mu }$&10&20&40&80&160\\
        \hline
        $\begin{array}{l}\quad \quad \quad {\rm{FR1}}\\(450{\rm{MHz}} - 5.9{\rm{GHz)}}\end{array}$&$\surd$&$\surd$&$\surd$&$\times$&$\times$\\
        \hline
        $\begin{array}{l}\quad \quad \quad {\rm{FR2}}\\(24.2{\rm{GHz}} - 52.6{\rm{GHz)}}\end{array}$&$\times$&$\times$&$\surd$&$\surd$&$\surd$\\
        \hline
        $T\left( {\mu s} \right)$&66.67&33.33&16.67&8.33&4.17\\
        \hline
        ${T_{{\rm{CP}}}}\left( {\mu s} \right)$&4.69&2.34&1.17&0.57&0.29\\
        \hline
        $T{\rm{ + }}{{T_{{\rm{CP}}}}}\left( {\mu s} \right)$&71.35&35.68&17.84&8.92&4.46\\
        \hline

        \end{tabular}
    \end{center}
\end{table}

In terms of sequence generation and time-frequency resources, PRS has unique advantages. Consider the downlink reference signals commonly used in 5G NR, namely SS, DMRS and CSI-RS, as comparisons, as shown in Table III. Note that the resource element (RE), defined as an OFDM symbol on a single subcarrier, is the smallest unit of time-frequency resource. SS contains primary synchronization signal (PSS) and secondary synchronization signal (SSS). As shown by Table III, the time-frequency resources of the SS and DMRS are limited. The CSI-RS also occupies fewer time-frequency resources than PRS due to the multiplexing of different antenna ports, while the PRS is rich in time-frequency resources. Therefore, the performance of PRS based radar sensing is expected to be better.
\begin{table*}[htbp]
    \renewcommand\arraystretch{1.2}
    \caption{Comparison of PRS with SS, DMRS and CSI-RS \textsuperscript{\cite{TR:CR22,TR:CR23}}}
    \begin{center}
        \begin{tabular}{|c|c|c|c|c|}
        \hline
        &SS&DMRS&CSI-RS&PRS\\
        \hline
        Signal Location&\makecell[l]{SS is located in the \\ Synchronization Broadcast \\ Block (SSB).}&\makecell[l]{DMRS exists in various \\ physical channels, including \\ downlink and uplink physical \\ channels.}&\makecell[l]{CSI-RS can only be sent on \\ downlink symbols, not on \\ the overlapped PRB with \\ SSB.}&\makecell[l]{The signal dedicated to \\ downlink positioning, which \\ cannot be mapped to the \\ resource particles allocated \\ to the SSB.}\\
        \hline
        \makecell[c]{Sequence \\ Generation}&\makecell[l]{PSS is a 127-length M \\ sequence, and SSS is a \\ 127-length Gold sequence.}&Gold sequence&Gold sequence&\makecell[l]{Gold sequence with a length \\ of 4096}\\
        \hline
        \makecell[c]{Time-frequency \\ Resources}&\makecell[l]{PSS and SSS use the first \\ and third symbol in an SSB \\ respectively, occupying 127 \\ subcarriers with the number \\ of sequences from 57 to 183 \\ among 144 REs in the \\ frequency domain.}&\makecell[l]{Front-load DMRS can occupy \\ 1/2 OFDM symbols in the \\ time domain. The DMRS of \\ physical broadcasting channel \\ (PBCH) occupies up to 20 \\ PRBs in the frequency \\ domain.}&\makecell[l]{CSI-RS can support up to \\ 32 different antenna ports, \\ which allows multiplexing \\ in one PRB. The time \\ domain can occupy 1/2/4 \\ OFDM symbols, and the \\ bandwidth can occupy up \\ to 52 PRBs.}&\makecell[l]{The frequency domain has a \\ comb-shaped pilot structure, \\ the bandwidth can occupy a \\ maximum of 272 PRBs, and \\ the time domain can occupy \\ multiple consecutive time \\ slots.}\\
        \hline

        \end{tabular}
    \end{center}
\end{table*}

PRS supports flexible time-frequency resource configuration to satisfy the positioning accuracy requirements in different application scenarios and avoid the waste of resources as well. According to the PRS resource mapping scheme defined in TS 38.211, PRS supports four comb forms, namely Comb 2/4/6/12, in the frequency domain, and four symbol number configurations, i.e., Symbol 2/4/6/12, in the time domain.

PRS supports the time-domain patterns summarized in Table IV with the RE offset $k_{{\rm{offset}}}^{{\rm{PRS}}} = 0$, where NA indicates that the mapping mode is not supported.
\begin{table}[htbp]
    \renewcommand\arraystretch{1.2}
    \caption{Time-domain mapping patterns supported by PRS
    \textsuperscript{\cite{TR:CR22}}}
    \begin{center}
        \begin{tabular}{|c|c|c|c|c|}
        \hline
        &2 symbols&4 symbols&6 symbols&12 symbols\\
        \hline
        Comb 2&0,1&0,1,0,1&0,1,0,1,0,1&\makecell[l]{0,1,0,1,0,1,0,1,0,1,\\0,1}\\
        \hline
        Comb 4&NA&0,2,1,3&NA&\makecell[l]{0,2,1,3,0,2,1,3,0,2,\\1,3}\\
        \hline
        Comb 6&NA&NA&0,3,1,4,2,5&\makecell[l]{0,3,1,4,2,5,0,3,1,4,\\2,5}\\
        \hline
        Comb 12&NA&NA&NA&\makecell[l]{0,6,3,9,1,7,4,10,2,\\8,5,11}\\
        \hline

        \end{tabular}
    \end{center}
\end{table}
Therefore, the two-dimensional time-frequency resource mapping diagrams of PRS in one time slot (14 OFDM symbols) and 12 consecutive subcarriers (one PRB) are obtained, as shown in Fig. \ref{fig_PRS time-frequency two-dimensional resource map}.
\begin{figure*}[htbp]
\centering
\includegraphics[width=0.7\textwidth]{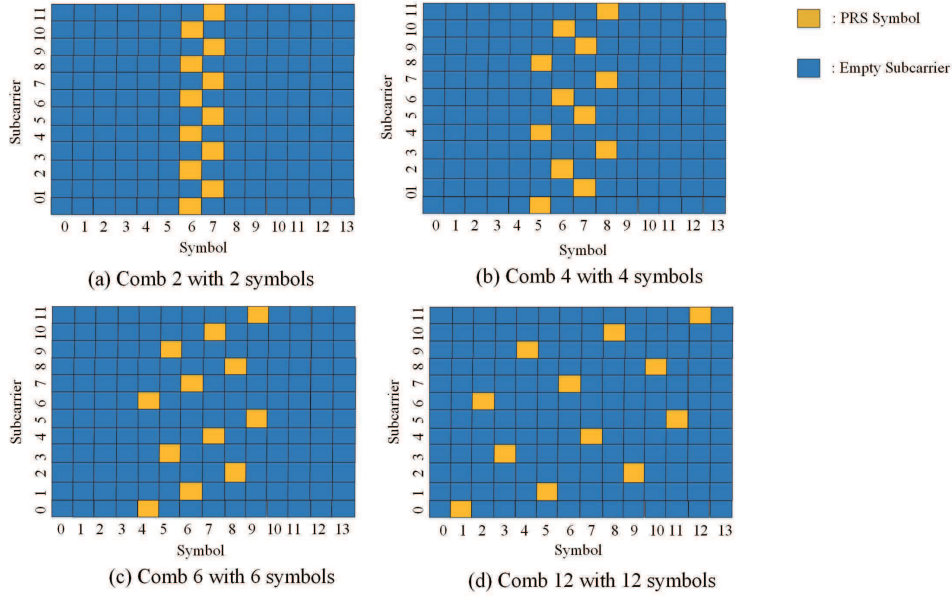}
\caption{The two-dimensional time-frequency resource map ping diagrams for PRS} \label{fig_PRS time-frequency two-dimensional resource map}
\end{figure*}

It can be summarized that the time-frequency resource allocation for PRS is flexible, so that multiple different downlink PRS signals from multiple base stations are multiplexed on different subcarriers in a comb-like manner. Thus, the comb structure of PRS is to control the interference of PRS signals transmitted by multiple BSs. For the radar sensing using PRS, it is also necessary to distinguish signals through different time-frequency structures to reduce interference.

\section{JSC Signal Model}
In this section, the PRS based JSC signal model is proposed. We consider the scenario where the 5G BS sends the downlink PRS to the target, and uses the echo signal from the target for range and velocity estimation. In addition, the ambiguity function is derived to further demonstrate its feasibility in radar sensing.
\subsection{Signal Model of PRS}
Taking Comb 4 with 4 symbols as an example, the PRS mapped OFDM signal model is shown in Fig. \ref{fig_Signal Model}. In a time slot, PRS occupies $M$ OFDM symbols in the time domain, $N$ subcarriers in the frequency domain, and the number of subcarriers carrying PRS is ${N_J}$ ($N/K_{{\rm{comb}}}^{{\rm{PRS}}}$), where $J$ is the set of subcarriers carrying PRS.
\begin{figure}[htbp]
\centering
\includegraphics[width=0.65\textwidth]{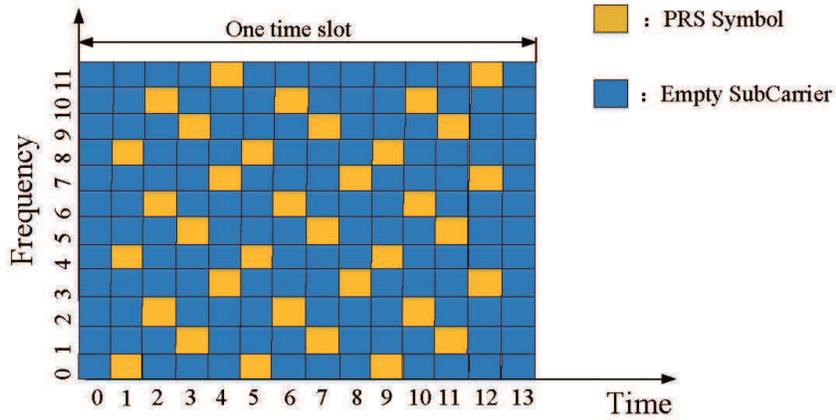}
\caption{PRS mapped OFDM signal model} \label{fig_Signal Model}
\end{figure}

The continuous time domain signal expression is
\begin{equation}
x\left( t \right) = \sum\limits_{m = 0}^{M - 1} {\sum\limits_{k = 0}^{{N_J} - 1} {s(k,m) \times {e^{j2\pi {f_k}t}}} {\rm{rect}}(\frac{{t - m{T_s}}}{{{T_s}}})},
\end{equation}
where $s\left( {k,m} \right)$ represents the modulated PRS symbol, with subcarrier index $k$ and OFDM symbol index $m$. ${T_s}$ is the total duration of the OFDM symbol which satisfies ${T_s}{\rm{ = }}T{\rm{ + }}{{T_{{\rm{CP}}}}}$. $\Delta f = 1/T$ is the frequency spacing of subcarriers, ${f_k}$ is the frequency of the $k$-th subcarriers that carry the PRS symbol. ${\rm{rect(}}t{\rm{/}}{T_s}{\rm{)}}$ is the rectangular function, which is equal to 1, for $0 \le t < {T_s}$, and 0, otherwise.

With the PRS equally spaced in frequency domain, the frequency of the subcarriers carrying PRS symbols ${f_k}$ satisfies the following equation.
\begin{equation}
{f_k} = (K_{{\rm{comb}}}^{{\rm{PRS}}} \times k + {k_0})\Delta f\;\quad k = 0, \ldots {N_J} - 1,
\end{equation}
where $K_{{\rm{comb}}}^{{\rm{PRS}}}$ is the comb size of PRS, ${k_0}$ is the index of the first subcarrier carrying PRS, which is related to $m$ and $K_{{\rm{comb}}}^{{\rm{PRS}}}$.

Taking Comb 4 with 4 symbols as an example, $K_{{\rm{comb}}}^{{\rm{PRS}}}=4$, and ${k_0}$ can take 0, 2, 1, 3 in turn. The relation between $k_0$ and $m$ is shown in (5)
\begin{equation}
{k_0} = \frac{{m\bmod K_{{\rm{comb}}}^{{\rm{PRS}}}}}{2} + \frac{3}{4}\left[ {1 - {{\left( { - 1} \right)}^{m\bmod K_{{\rm{comb}}}^{{\rm{PRS}}}}}} \right],
\end{equation}
where mod refers to a modulo operation.
\subsection{Received Radar Signal}
According to radar signal processing algorithm in \cite{TR:CR10}, the echo signal received by radar receiver is
\begin{equation}
\begin{array}{l}
{{S_{{\rm{Rx}}}}}\left[ {k,m} \right] = \xi {{S_{{\rm{Tx}}}}}\left[ {k,m} \right] \times {e^{-j2\pi \left( {K_{{\rm{comb}}}^{{\rm{PRS}}}k\Delta f} \right)\frac{{2{R_r}}}{c}}}\\
 \times {e^{j2\pi \left( {K_{{\rm{comb}}}^{{\rm{PRS}}}m{T_s}} \right){f_{d,r}}}}
\end{array},
\end{equation}
where ${{S_{{\rm{Tx}}}}}\left[ {k,m} \right]$, ${{S_{{\rm{Rx}}}}}\left[ {k,m} \right]$ are the transmitted and the received modulation symbols respectively, $\xi $ is the attenuation factor, which is regarded as constant during the transmission of PRS. ${R_r}$ is the range of target, ${f_{d,r}}$ is the Doppler shift, and $c$ is the speed of light.

Dividing the received modulation symbols ${{S_{{\rm{Rx}}}}}\left[ {k,m} \right]$ by the transmitted  modulation symbols ${{S_{{\rm{Tx}}}}}\left[ {k,m} \right]$, the following matrix  ${\left( {{S_g}} \right)_{k,m}}$ is obtained.
\begin{equation}
{\left( {{S_g}} \right)_{k,m}} = \frac{{{{S_{{\rm{Rx}}}}}\left[ {k,m} \right]}}{{{{S_{{\rm{Tx}}}}}\left[ {k,m} \right]}} = \xi \left( {{{\bar k}_r} \otimes {{\bar k}_d}} \right),
\end{equation}
where ${\bar k_r}\left( k \right) = {e^{ - j2\pi \left( {K_{{\rm{comb}}}^{{\rm{PRS}}}k\Delta f} \right)\frac{{2{R_r}}}{c}}}$ and ${\bar k_d}\left( m \right) = {e^{j2\pi \left( {K_{{\rm{comb}}}^{{\rm{PRS}}}m{T_s}} \right){f_{d,r}}}}$ are the two vectors carrying the range and the Doppler information. $\otimes$ refers to a dyadic product.

\subsection{Ambiguity Function of PRS}
The definition of ambiguity function is
\begin{equation}
\chi \left( {\tau ,{f_{d,r}}} \right) = \int\limits_{ - \infty }^\infty  {x\left( t \right)} {x^*}\left( {t - \tau } \right){e^{j2\pi {f_{d,r}}t}}dt,
\end{equation}
where ${x^*}( \cdot )$ is a conjugate operation, $\tau$ is the round-trip delay experienced by signal transmission.

Substituting the signal expression of (3) into (8), the following expression can be obtained.
\begin{equation}
\begin{array}{l}
\chi \left( {\tau ,{f_{d,r}}} \right) = \sum\limits_{{m_1} = 0}^{M - 1} {\sum\limits_{{m_2} = 0}^{M - 1} {\sum\limits_{{k_1} = 0}^{{N_J} - 1} {\sum\limits_{{k_2} = 0}^{{N_J} - 1} {s\left( {{k_1},{m_1}} \right){s^*}\left( {{k_2},{m_2}} \right)} } } } \\
 \cdot {e^{j2\pi {f_{{k_2}}}\tau }} \cdot \int\limits_{ - \infty }^{ + \infty } {{e^{j2\pi \left( {{f_{{k_1}}} - {f_{{k_2}}} + {f_{d,r}}} \right)t}}} {\rm{rect}}(\frac{{t - {m_1}{T_s}}}{{{T_s}}})\\
 \times {\rm{rect}}(\frac{{t - \tau  - {m_2}{T_s}}}{{{T_s}}})dt
\end{array}.
\end{equation}

The ambiguity function of PRS can be abbreviated as
\begin{equation}
\begin{array}{l}
\chi \left( {\tau ,{f_{d,r}}} \right) = \sum\limits_{{m_1} = 0}^{M - 1} {\sum\limits_{{m_2} = 0}^{M - 1} {\sum\limits_{{k_1} = 0}^{{N_J} - 1} {\sum\limits_{{k_2} = 0}^{{N_J} - 1} {s\left( {{k_1},{m_1}} \right){s^*}\left( {{k_2},{m_2}} \right)} } } } \\
 \cdot {e^{j2\pi {f_{{k_2}}}\tau }}\int\limits_{{t_{\min }}}^{{t_{\max }}} {{e^{j2\pi \left( {{f_{{k_1}}} - {f_{{k_2}}} + {f_{d,r}}} \right)t}}} dt\\
 = {t_{\max  - \min }}\sum\limits_{{m_1} = 0}^{M - 1} {\sum\limits_{{m_2} = 0}^{M - 1} {\sum\limits_{{k_1} = 0}^{{N_J} - 1} {\sum\limits_{{k_2} = 0}^{{N_J} - 1} {s\left( {{k_1},{m_1}} \right){s^*}\left( {{k_2},{m_2}} \right)} } } } \\
 \cdot {e^{j2\pi {f_{{k_2}}}\tau }} \cdot \rm{sinc}\left( {\left( {{f_{{k_1}}} - {f_{{k_2}}} + {f_{d,r}}} \right){t_{\max  - \min }}} \right)\\
 \cdot {e^{j2\pi \left( {{f_{{k_1}}} - {f_{{k_2}}} + {f_{d,r}}} \right){t_{avg}}}}
\end{array},
\end{equation}
where $\rm{sinc}\left(  \cdot  \right)$ refers to a normalized sinc function, the definitions of ${t_{\max }}$, ${t_{\min }}$ , ${t_{\max  - \min }}$ and ${t_{\rm{avg}}}$ are as follows.
\begin{equation}
\begin{array}{l}
{t_{\max }} = \min \left\{ {\left( {{m_1} + 1} \right){T_s},\tau  + \left( {{m_2} + 1} \right){T_s}} \right\},\\
{t_{\min }} = \max \left\{ {{m_1}{T_s},\tau  + {m_2}{T_s}} \right\},\\
{t_{\max  - \min }} = {t_{\max }} - {t_{\min }},\\
{t_{\rm{avg}}} = \frac{{{t_{\max }} + {t_{\min }}}}{2}.
\end{array}
\end{equation}

It can be discovered from Fig. \ref{fig_mohu} that the ambiguity function of PRS is ``pushpin type" which can provide high range and velocity resolution at the same time. Thus, PRS has the potential to realize radar sensing.
\begin{figure}[!t]
\centering
\includegraphics[width=0.65\textwidth]{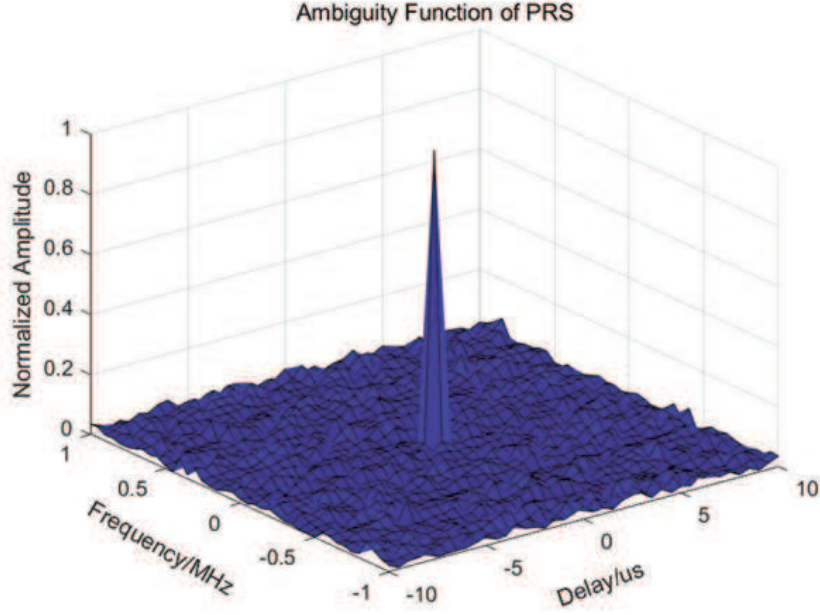}
\caption{Ambiguity function of PRS (Comb 4 with 4 symbols, $N$=1024, ${N_J}{\rm{ = }}256$)} \label{fig_mohu}
\end{figure}

\section{PRS Based Radar Sensing}
In this section, the radar signal processing algorithm in \cite{TR:CR10} is adopted. Then, the resolution, the maximum
unambiguous range/velocity, and the accuracy of PRS based radar sensing are analyzed. Besides, the methods improving the accuracy of radar sensing are proposed.
\subsection{Range Estimation using PRS}
${N_J}$-point inverse fast Fourier transform (IFFT) is performed on the PRS, which is inserted at equal intervals in the frequency domain. IFFT is performed on the $m$-th column of the ${\left( {{S_g}} \right)_{k,m}}$
\begin{equation}
\begin{array}{l}
r\left( l \right) = {\rm{IFFT}}\left( {{S_{g,m}}(k)} \right) = \sum\limits_{k = 0}^{{N_J} - 1} {{e^{ - j2\pi K_{{\rm{comb}}}^{{\rm{PRS}}}\Delta fk\frac{{2{R_r}}}{c}}}}  \times {e^{j2\pi lk/{N_J}}}\\
l \in \left\{ {0,1, \ldots {N_J} - 1} \right\},
\end{array}
\end{equation}
where ${\rm{IFFT}}\left(  \cdot  \right)$ is a inverse fast Fourier transform operation. $l$ is the index of the $l$-th result of $N_J$-point IFFT. The peak occurs when the two exponential terms in (12) cancel each other. With the index of the peak of column $m$ recorded as $in{d_{{S_g},m}}$, the estimated range ${{\hat R}_r}$ is derived as
\begin{equation}
\begin{array}{l}
{{\hat R}_r} = \frac{{c\hat l}}{{2{N_J}K_{{\rm{comb}}}^{{\rm{PRS}}}\Delta f}}\\
\quad \, = \frac{{in{d_{{S_g},m}}c}}{{2N\Delta f}}\quad \;in{d_{{S_g},m}} \in \left\{ {0,1, \ldots {N_J} - 1} \right\}.
\end{array}
\end{equation}

The range resolution $\Delta R$ is derived as
\begin{equation}
\Delta R{\rm{ = }}\frac{c}{{2N\Delta f}}\;.
\end{equation}

Limited by the maximum value of ${in{d_{{S_g},m}}}$, the maximum unambiguous range ${R_{\max }}$ can be expressed as (15)
\begin{equation}
{R_{\max }} = \frac{{{N_J}c}}{{2N\Delta f}}{\rm{ = }}\frac{c}{{2K_{{\rm{comb}}}^{{\rm{PRS}}}\Delta f}}.
\end{equation}

(15) shows that short-range, medium-range, and long-range ranging can be realized by tuning the comb size of PRS ${K_{{\rm{comb}}}^{{\rm{PRS}}}}$. However, the restriction imposed by the guard interval is usually more stringent in actual target detection.

Averaging the estimated results of all columns of ${\left( {{S_g}} \right)_{k,m}}$ yields an estimate of the range $\hat R_r$. With ${N_t}$ measurements, the root mean square error (RMSE) of the range estimation is derived as
\begin{equation}
RMSE\left( {{R_r}} \right){\rm{ = }}\sqrt {\frac{1}{{{N_t}}}\sum\limits_{i = 1}^{N_t} {{{\left( {{{\hat R}_r}\left[ i \right] - {R_r}} \right)}^2}} }.
\end{equation}

According to \cite{TR:CR21}, the method of increasing the number of IFFT points can be adopted to improve the ranging accuracy. By introducing the fractional factor ${m_a}$, the estimation function using fractional Fourier transform (FRFT) is
\begin{equation}
\begin{array}{l}
{\rm{IFFT}}\left( {{S_{g,m}}(k)} \right) = \frac{1}{{{m_a} \times {N_J}}}\sum\limits_{k = 0}^{{m_a} \times {N_J} - 1} {{e^{ - j2\pi K_{{\rm{comb}}}^{{\rm{PRS}}}\Delta fk\frac{{2{R_r}}}{c}}}} \\
 \times {e^{\frac{{j2\pi lk}}{{{m_a} \times {N_J}}}}}
\end{array}.
\end{equation}

Meanwhile, the range estimation in (13) can be rewritten as
\begin{equation}
{\hat R_{r,{m_a}}} = \frac{{ind_{{S_g},m}^{{m_a}}c}}{{2 \times {m_a} \times N\Delta f}}\;\;ind_{{S_g},m}^{{m_a}} \in \left\{ {0,1, \ldots {m_a} \times {N_J} - 1} \right\},
\end{equation}
where $ind_{{S_g},m}^{{m_a}}$ is the index of the peak for IFFT after introducing the fractional factor ${m_a}$.

The fractional factor increases more search steps. Due to the fine accuracy of the estimation curve, the ranging accuracy is improved \cite{TR:CR21}. However, the ranging accuracy will only approach CRLB with infinitely increasing ${m_a}$. Since the complexity of IFFT is $o\left( {{m_a} \times N \cdot {{\log }_2}({m_a}N}) \right)$, the tradeoff between accuracy and complexity when choosing $m_a$ needs to be considered. However, the maximum unambiguous range does not change with the increase of the IFFT points.
\subsection{Velocity Estimation using PRS}
Similar to the range estimation, ${M_J}$ ($M/K_{{\rm{comb}}}^{{\rm{PRS}}}$)-points fast Fourier transform (FFT) is performed on the ${M_J}$ PRS symbols, which is inserted at equal intervals in the time domain. FFT is performed on the $k$-th row of the ${\left( {{S_g}} \right)_{k,m}}$
\begin{equation}
\begin{array}{l}
v\left( d \right) = {\rm{FFT}}\left( {{S_{g,k}}(m)} \right) = \sum\limits_{m = 0}^{{M_J} - 1} {{e^{j2\pi \left( {K_{{\rm{comb}}}^{{\rm{PRS}}}m{T_s}} \right){f_{d,r}}}}} \\
 \times {e^{ - j2\pi dm/{M_J}}}\,\,\;\quad d \in \left\{ {0,1, \ldots {M_J} - 1} \right\},
\end{array}
\end{equation}
where ${\rm{FFT}}\left(  \cdot  \right)$ is a fast Fourier transform operation. $d$ is the index of the $d$-th result of $M_J$-point FFT. Recording the peak index of row $k$ as $in{d_{{S_g},k}}$, the estimation of Doppler frequency shift ${\hat f_{d,r}}$ is:
\begin{equation}
{\hat f_{d,r}} = \frac{{in{d_{{S_g},k}}}}{{M{T_s}}}.
\end{equation}

According to the relation between velocity $v$ and Doppler frequency shift ${\hat f_{d,r}}$
\begin{equation}
v = \frac{{c{f_{d,r}}}}{{2{f_c}}},
\end{equation}
where $f_c$ is the carrier frequency of the transmitted signal, the estimation of velocity $\hat v$ is
\begin{equation}
\hat v = \frac{{in{d_{{S_g},k}}c}}{{2M{T_s}{f_c}}}\quad in{d_{{S_g},k}} \in \left\{ {0,1, \ldots {M_J} - 1} \right\}.
\end{equation}

The velocity resolution $\Delta v$ is shown in (23).
\begin{equation}
\Delta v{\rm{ = }}\frac{c}{{2M{T_s}{f_c}}},
\end{equation}
and the maximum unambiguous velocity ${v_{\max }}$ is
\begin{equation}
{v_{\max }} = \frac{{{M_J}c}}{{2M{T_s}{f_c}}}{\rm{ = }}\frac{c}{{2K_{{\rm{comb}}}^{{\rm{PRS}}}{T_s}{f_c}}}.
\end{equation}

All rows are averaged to provide an estimate of velocity $\hat v$. The RMSE of the velocity estimation is
\begin{equation}
RMSE\left( v \right){\rm{ = }}\sqrt {\frac{1}{{{N_t}}}\sum\limits_{i = 1}^{N_t} {{{\left( {\hat v\left[ i \right] - v} \right)}^2}} }.
\end{equation}

The method of increasing the number of FFT points can also be adopted to improve the accuracy of velocity measurement, thereby reducing the number of PRS symbols and the overhead of the required time slot. The maximum unambiguous velocity does not change with the increase of the FFT points.

However, the increase of FFT points will bring computational complexity. And the velocity measurement using single frame will consume too many time slots within a frame. Thus, the velocity measurement using multiple frames is proposed to reduce the overhead in each frame. The details are shown in Fig. \ref{fig_multiframe}.
\begin{figure}[!t]
\centering
\includegraphics[width=0.6\textwidth]{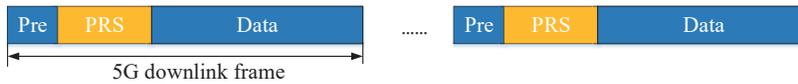}
\caption{Velocity measurement using multiple frames (Pre is the abbreviation of preamble sequence)} \label{fig_multiframe}
\end{figure}

${N_f}$ 5G downlink frames are continuously sent. The time domain resources occupied by PRS in the $i$-th downlink frame contain ${S_i}{\kern 1pt} {\kern 1pt} (i{\rm{ = }}1,2 \ldots {N_f})$ symbols. The receiver extracts the PRS symbols in each frame for velocity estimation
\begin{equation}
{\rm{FFT}}\left( {{S_{g,k}}(m)} \right) = \sum\limits_{m = 0}^{{S_J} - 1} {{e^{j2\pi \left( {K_{{\rm{comb}}}^{{\rm{PRS}}}m{T_s}} \right){f_{d,r}}}}}  \times {e^{ - j2\pi dm/{S_J}}},
\end{equation}
where ${S_J}{\rm{ = }}\sum\limits_{i = 1}^{{N_f}} {{S_i}} /K_{{\rm{comb}}}^{{\rm{PRS}}}$ is the PRS symbols inserted at equal intervals in the time domain which are extracted from $N_f$ frames.

The estimation of velocity ${\hat v_{{\rm{multi}}}}$ is
\begin{equation}
\begin{array}{l}
{{\hat v_{{\rm{multi}}}}} = \frac{{ind_{{S_g},k}^{{\rm{multi}}}c}}{{2{T_s}{f_c}(\sum\limits_{i = 1}^{{N_f}} {{S_i}})}}\\
ind_{{S_g},k}^{{\rm{multi}}} \in \{ 0,1, \ldots {S_J} - 1\}
\end{array},
\end{equation}
where $ind_{{S_g},k}^{{\rm{multi}}}$ is the peak index of velocity measurement using multiple frames.

The velocity resolution $\Delta {v_{{\rm{multi}}}}$ is derived as
\begin{equation}
\Delta {v_{{\rm{multi}}}} = \frac{c}{{2{T_s}{f_c}(\sum\limits_{i = 1}^{{N_f}} {{S_i}})}}.
\end{equation}

PRS symbol overhead ${\eta _i}$ in the $i$-th frame is defined as
\begin{equation}
{\eta _i}{\rm{ = }}\frac{{{S_i}}}{N_{{\rm{slot}}}^{{\rm{frame,}}\mu } \cdot {N_{{\rm{symb}}}^{{\rm{slot}}}}}.
\end{equation}

In addition, we define the sensing refresh time $\rho$ to represent the time to achieve a velocity measurement
\begin{equation}
\rho {\rm{ = }}{N_f} \cdot {T_f},
\end{equation}
where ${T_f} = 10\;{\rm{ms}}$ is the duration of per frame.

For the velocity measurement using multiple frames, the total number of PRS symbols used for sensing is increased compared with a single frame, thereby improving the velocity resolution $\Delta {v_{{\rm{multi}}}}$. Thus, the accuracy of velocity estimation is improved. Besides, the overhead within a frame ${\eta _i}$ is reduced, with the cost of increasing the sensing refresh time $\rho $ which indicates the time to achieve a velocity measurement becomes longer.

Referring to the scenario of JSC roadside BS detecing vehicles shown in Fig. 1, the long sensing refresh time means that the vehicle has to travel a long range to achieve a velocity measurement, which may cause the vehicle velocity to change during this period, namely, the velocity estimation is not updated in time. Therefore, there is a performance tradeoff among the velocity resolution $\Delta {v_{{\rm{multi}}}}$, overhead in the $i$-th frame ${\eta _i}$ and sensing refresh time $\rho$.

In addition, PRS is sent in the 5G downlink frame for sensing, so that sensing and communication functions are time-division multiplexed. The impact of velocity measurement using multiple frames on communication performance is mainly reflected in the symbol overhead. Assuming that ${M_d}$ data symbols are transmitted in a single frame and $P$-QAM modulation is adopted, the maximum communication data rate ${R_d}$ of a single frame is
\begin{equation}
{R_d}{\rm{ = }}{M_d} \times N \times {\log _2}P/{T_f},
\end{equation}

The reduction of PRS symbol overhead ${\eta _i}$ indicates that the amount of resources used for communication in a frame increases, which improves the communication rate.
\section{CRLB of Radar Sensing and Positioning}
In this section, we derive the CRLBs of  PRS based radar sensing and positioning.
\subsection{CRLB of PRS based Radar Sensing}
\begin{Theo}\label{th_e_f}
When $M \gg 1,N \gg 1$, the CRLB of the range and velocity estimated by the PRS based radar sensing is
\begin{equation}
CRLB\left( R_r \right){\rm{ = }}\frac{{{c^2}{T^2}}}{{{\xi ^2}SNR{{\left( {2\pi } \right)}^2}}}\frac{{12}}{{MN\left( {{N_J} - 1} \right)\left( {7{N_J} + 1} \right)}},
\end{equation}
\begin{equation}
CRLB\left( v \right){\rm{ = }}\frac{{{c^2}}}{{{\xi ^2}SNR{{\left( {2\pi } \right)}^2}f_c^2T_s^2}}\frac{{12}}{{NM\left( {{M_J} - 1} \right)\left( {7{M_J} + 1} \right)}},
\end{equation}
where $SNR$ is signal-to-noise ratio.
\end{Theo}
\begin{proof}
For the traditional OFDM system, since the radar receiver has known modulation symbols and the noise obeys one-dimensional Gaussian distribution, the radar observations after phase-by-phase rotation \cite{TR:CR25} are expressed as
\begin{equation}
{z_{m,n}} = \xi {A_{m,n}}{e^{j2\pi K_{{\rm{comb}}}^{{\rm{PRS}}}m{T_s}{f_{d,r}}}}{e^{ - j2\pi nK_{{\rm{comb}}}^{{\rm{PRS}}}\Delta f\tau }} + {w_{m,n}},
\end{equation}
where ${A_{m,n}} = \left| {{x_{m,n}}} \right|$ is the amplitude of the transmitted symbol. ${w_{m,n}}$ is the additive white Gaussian noise (AWGN) with ${\sigma ^2}$ variance and zero mean.

The received signal with unknown parameter $\theta {\rm{ = }}\left( {\tau ,{f_{d,r}}} \right)$ is observed, then the likelihood function is
\begin{equation}
\begin{array}{l}
f\left( {\left. z \right|\tau ,{f_{d,r}}} \right) = \frac{1}{{{{\left( {2\pi {\sigma ^2}} \right)}^{\frac{{MN}}{2}}}}}\\
{e^{ - \frac{1}{{2{\sigma ^2}}}\sum\limits_m {\sum\limits_n {{{\left| {{z_{m,n}} - \xi {A_{m,n}}{e^{j2\pi K_{{\rm{comb}}}^{{\rm{PRS}}}m{T_s}{f_{d,r}}}}{e^{ - j2\pi nK_{{\rm{comb}}}^{{\rm{PRS}}}\Delta f\tau }}} \right|}^2}} } }}
\end{array}.
\end{equation}

The log likelihood function $L\left( {\left. z \right|\tau ,{f_{d,r}}} \right)$ is derived as
\begin{equation}
\begin{array}{l}
L\left( {\left. z \right|\tau ,{f_{d,r}}} \right){\rm{ = }}\ln f\left( {\left. z \right|\tau ,{f_{d,r}}} \right) =  - \frac{{MN}}{2}\ln \left( {2\pi {\sigma ^2}} \right)\\
 - \frac{1}{{2{\sigma ^2}}}\sum\limits_m {\sum\limits_n {{{\left| {{z_{m,n}} - \xi {A_{m,n}}{e^{j2\pi K_{{\rm{comb}}}^{{\rm{PRS}}}m{T_s}{f_{d,r}}}}{e^{ - j2\pi nK_{{\rm{comb}}}^{{\rm{PRS}}}\Delta f\tau }}} \right|}^2}} }
\end{array},
\end{equation}
with ${s_{m,n}} = \xi {A_{m,n}}{e^{j2\pi K_{{\rm{comb}}}^{{\rm{PRS}}}m{T_s}{f_{d,r}}}}{e^{ - j2\pi nK_{{\rm{comb}}}^{{\rm{PRS}}}\Delta f\tau }}$. (35) can be simplified as
\begin{equation}
L\left( {\left. z \right|\tau ,{f_{d,r}}} \right){\rm{ = }} - \frac{{MN}}{2}\ln \left( {2\pi {\sigma ^2}} \right) - \frac{1}{{2{\sigma ^2}}}\sum\limits_m {\sum\limits_n {{{\left| {{z_{m,n}} - {s_{m,n}}} \right|}^2}} }.
\end{equation}

Further, we can obtain the second-order Fisher information matrix $F$ as follows.
\begin{equation}
\begin{array}{l}
F{\rm{ = }}\left[ {\begin{array}{*{20}{c}}
{{F_{\tau \tau }}}&{{F_{\tau {f_{d,r}}}}}\\
{{F_{{f_{d,r}}\tau }}}&{{F_{{f_{d,r}}{f_{d,r}}}}}
\end{array}} \right]\\
 =  - \left[ {\begin{array}{*{20}{c}}
{E\left( {\frac{{{\partial ^2}L}}{{\partial {\tau ^2}}}} \right)}&{E\left( {\frac{{{\partial ^2}L}}{{\partial \tau \partial {f_{d,r}}}}} \right)}\\
{E\left( {\frac{{{\partial ^2}L}}{{\partial {f_{d,r}}\partial \tau }}} \right)}&{E\left( {\frac{{{\partial ^2}L}}{{\partial {f_{d,r}}^2}}} \right)}
\end{array}} \right]\\
{\rm{ = }}\frac{1}{{{\sigma ^2}}}{\left( {\sum\limits_m {\sum\limits_n {\left[ {\frac{{\partial {s_{m,n}}}}{{\partial {\theta _i}}} \cdot \frac{{\partial s_{m,n}^*}}{{\partial {\theta _j}}}} \right]} } } \right)_{ij}}.
\end{array}
\end{equation}

The CRLB matrix for time delay and Doppler shift estimation is the inverse of the Fisher information matrix as follows.
\begin{equation}
\left[ {\begin{array}{*{20}{c}}
{CRLB\left( \tau  \right)}&{CRLB\left( {\tau ,{f_{d,r}}} \right)}\\
{CRLB\left( {{f_{d,r}},\tau } \right)}&{CRLB\left( {{f_{d,r}}} \right)}
\end{array}} \right]{\rm{ = }}{F^{ - 1}}.
\end{equation}

Assuming $M \gg 1,N \gg 1$ , the CRLB of range and velocity can be obtained.
\begin{equation}
{\begin{array}{*{20}{l}}
{CRLB\left( \tau  \right){\rm{ = }}\frac{{{F_{{f_{d,r}}{f_{d,r}}}}}}{{{F_{\tau \tau }}{F_{{f_{d,r}}{f_{d,r}}}} - {F_{\tau {f_{d,r}}}}{F_{{f_{d,r}}\tau }}}}}\\
{ = \frac{{{T^2}}}{{{\xi ^2}SNR \cdot {{\left( {2\pi } \right)}^2}}} \cdot {\mkern 1mu} \frac{{48}}{{MN\left( {{N_J} - 1} \right)\left( {7{N_J} + 1} \right)}}}
\end{array}},
\end{equation}
\begin{equation}
\begin{array}{*{20}{l}}
{CRLB\left( {{f_{d,r}}} \right){\rm{ = }}\frac{{{F_{\tau \tau }}}}{{{F_{\tau \tau }}{F_{{f_{d,r}}{f_{d,r}}}} - {F_{\tau {f_{d,r}}}}{F_{{f_{d,r}}\tau }}}}}\\
{{\rm{ = }}\frac{1}{{{\xi ^2}SNR \cdot {{\left( {2\pi } \right)}^2} \cdot T_s^2}} \cdot \frac{{48}}{{NM\left( {{M_J} - 1} \right)\left( {7{M_J} + 1} \right)}}}
\end{array}.
\end{equation}

According to the relation ${R_r} = c\tau /2$ and (21), (31) and (32) can be derived based on the following relation
\begin{equation}
CRLB\left( R_r \right) = \frac{{{c^2}}}{4}CRLB\left( \tau  \right),
\end{equation}
\begin{equation}
CRLB\left( v \right) = \frac{{{c^2}}}{{4f_c^2}}CRLB\left( {{f_{d,r}}} \right).
\end{equation}
\end{proof}

(32-33) show that with the decreasing of $T$, the CRLB of range estimation decreases, while the CRLB of velocity estimation increases. There is a performance tradeoff between range and velocity estimation. Therefore, the parameters can be reasonably configured to satisfy the accuracy requirements for range and velocity estimation in actual radar sensing.
\subsection{CRLB of PRS based Positioning}
\begin{Theo}\label{th_e_s}
The CRLB for PRS positioning is
\begin{equation}
CRLB\left( R_r \right){\rm{ = }}\frac{{{c^2}T^2}}{{\frac{{4{\pi ^2}}}{3} \cdot SNR \cdot N\left( {{N_J} - 1} \right)\left( {2{N_J} - 1} \right)}}.
\end{equation}
\end{Theo}
\begin{proof}
The BS generates the PRS according to the high-level parameter configuration, and maps the PRS to the physical resource unit, and finally obtains the OFDM baseband signal $x\left[ n \right]$ using one symbol through IFFT.
\begin{equation}
\begin{array}{l}
x\left[ n \right] = \sqrt {\frac{C}{N}} \sum\limits_{l \in J} {{p_l} \cdot {s_l} \cdot {e^{\frac{{j2\pi nl}}{N}}}} \\
{\rm{ = }}\sqrt {\frac{C}{N}} {e^{\frac{{j2\pi n{k_0}}}{N}}}\sum\limits_{k = 0}^{{N_J} - 1} {{p_k} \cdot {s_k} \cdot \,} {e^{\frac{{j2\pi nk}}{{{N_J}}}}}
\end{array},
\end{equation}
where $C$ is the signal power, and ${s_k}$ is the PRS. $k_0$ is the index of the first subcarrier carrying PRS given in (4). $p_k^2$ is the relative power weight of subcarrier $k$ with $\sum\nolimits_k {p_k^2}  = N$. Assuming that the power is evenly distributed on all subcarriers carrying PRS, the relative power weight of subcarrier $k$ is
\begin{equation}
{p_k} = \sqrt {\frac{N}{{{N_J}}}} {\rm{ = }}\sqrt {K_{{\rm{comb}}}^{{\rm{PRS}}}} ,\quad \;k \in J.
\end{equation}

Therefore, according to general CRLB derivation results of signals in gaussian white noise given by Kay \cite{TR:CR26} and the derivation in \cite{TR:CR27}, the CRLB estimation of the delay for PRS positioning is
\begin{equation}
CRLB\left( \tau  \right){\rm{ = }}\frac{{T^2}}{{\frac{{4{\pi ^2}}}{3} \cdot SNR \cdot N\left( {{N_J} - 1} \right)\left( {2{N_J} - 1} \right)}}.
\end{equation}

Thus, the range estimation of CRLB in (44) can be obtained.
\end{proof}
\section{Simulations and Numerical results}
In this section, we first provide simulation results of radar signal processing method to verify the feasibility of PRS based radar sensing in range and velocity estimation. Then, the performance of range and velocity estimation is provided and compared with CRLB. In addition, we give some suggestions for the frame structure design for 5G-A and 6G.
\subsection{Range Estimation}
For the OFDM waveform modulated by PRS sequence, one time slot accounts for up to 12 symbols. The $SNR$ is 5 dB. The rest of simulation parameters are shown in Table V.
\begin{table}[htbp]
    \renewcommand\arraystretch{1.5}
    \caption{simulation parameters \textsuperscript{\cite{TR:CR22}}}
    \begin{center}
        \begin{tabular}{|m{2cm}<{\centering}|m{3cm}<{\centering}|m{2cm}<{\centering}|}
        \hline
        Symbol&Parameter&Value\\
        \hline
        ${f_c}$&Carrier frequency&$24\;{\rm{GHz}}$\\
        \hline
        $\Delta f$&Frequency spacing of subcarriers&$120\;{\rm{kHz}}$\\
        \hline
        $T_s$&Total duration of OFDM symbol&$8.92\;\mu s$\\
        \hline
        $N$&Number of subcarriers&256\\
        \hline
        $K_{{\rm{comb}}}^{{\rm{PRS}}}$&Comb size of PRS&4/2\\
        \hline
        $\begin{array}{l}\;\quad\;\;{N_J}\;\\(N/K_{{\rm{comb}}}^{{\rm{PRS}}})\end{array}$&Number of subcarriers carrying PRS&64/128\\
        \hline
        ${R_r}$&Range of target&$50\;{\rm{m}}$\\
        \hline
        $v$&Velocity of target&$15\;{\rm{m}}/{\rm{s}}$\\
        \hline

        \end{tabular}
    \end{center}
\end{table}

A measurement result is shown in Fig. \ref{fig_rangepeak}, where the theoretical value of range resolution is $\Delta R = \frac{c}{{2N\Delta f}} = 4.88\;{\rm{m}}$. Since the peak index of column IFFT in Fig. \ref{fig_rangepeak} is 11 and the count starts from 1 in the simulation, the value of $in{d_{{S_g},m}}$ is 10. The estimated range is
\begin{figure}[!t]
\centering
\includegraphics[width=0.45\textwidth]{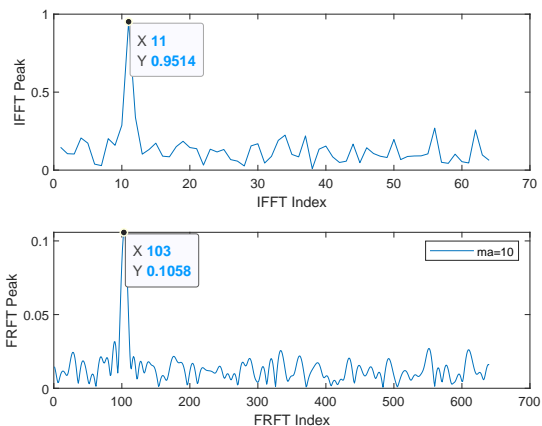}
\caption{Simulated IFFT correlation peak when the target is at 50 m} \label{fig_rangepeak}
\end{figure}
\begin{equation}
{\hat R_r} = \frac{{in{d_{{S_g},m}}c}}{{2N\Delta f}} = 48.83\;{\rm{m}}.
\end{equation}

With the fractional factor ${m_a} = 10$, the value of $ind_{{S_g},m}^{{m_a}}$ is 102 and the estimated range is
\begin{equation}
{\hat R_{r,{m_a}}} = \frac{{ind_{{S_g},m}^{{m_a}}c}}{{2 \times {m_a} \times N\Delta f}} = 49.80\;{\rm{m}}.
\end{equation}

With 1000 times Monte Carlo simulations, the RMSE of range estimation is 1.17 m. When introducing the fractional factor, the RMSE is reduced to 0.33 m. From the results of range estimation, the ranging accuracy is significantly improved and reached the centimeter level.

\begin{figure}[!t]
\centering
\includegraphics[width=0.45\textwidth]{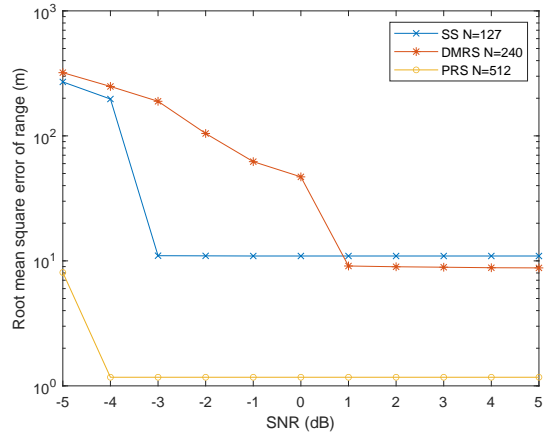}
\caption{Comparison of RMSEs of range estimations between SS and PRS} \label{fig_Compare}
\end{figure}

Since the CSI-RS occupies very few REs in one PRB due to the multiplexing of different antenna ports, we only compare the RMSE of range estimation among SS, DMRS and PRS to illustrate the feasibility of PRS for radar sensing, as shown in Fig. \ref{fig_Compare}. The ranging accuracy of the radar signal processing algorithm in this paper is limited by the number of subcarriers. Therefore, the comparison of the sensing performance of different reference signals essentially lies in their different time-frequency resources. The comb of PBCH DMRS in the frequency domain, occupies up to 240 consecutive subcarriers. For a reasonable comparison, the length of PRS sequence is set the same as the synchronization sequence and the comb size is 4 in Fig. \ref{fig_Compare}. Since the frequency domain is comb-shaped, PRS occupies more frequency resources than SS. Hence, the ranging performance of PRS is significantly better than that of SS. Although both PRS and DMRS are mapped with Comb 4 in the frequency domain, the PRS sequence is longer, so the ranging accuracy is better. The result shows that PRS is suitable for radar sensing.

\begin{figure}[!t]
\centering
\includegraphics[width=0.45\textwidth]{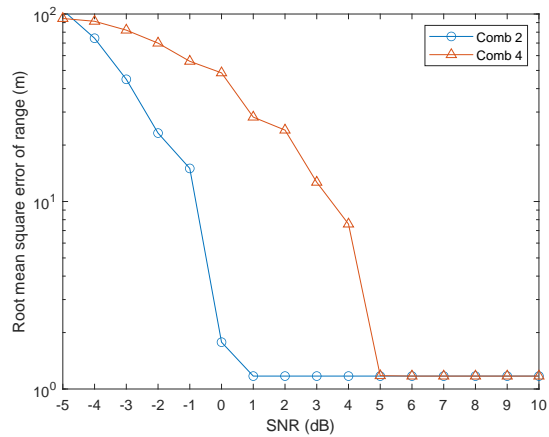}
\caption{Comparison of RMSEs of range estimations between Comb 2 and Comb 4} \label{fig_comb}
\end{figure}

\begin{figure}[!t]
\centering
\includegraphics[width=0.45\textwidth]{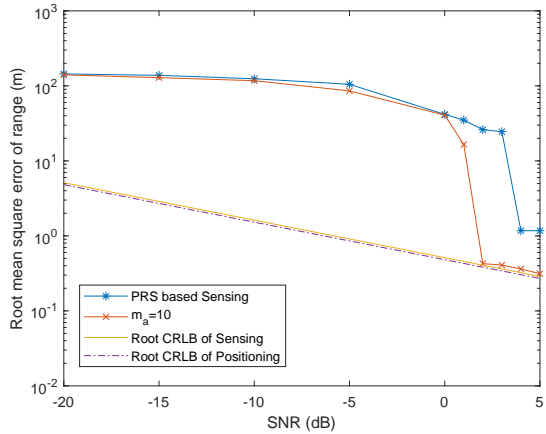}
\caption{RMSEs of range estimations vs SNR} \label{fig_Range}
\end{figure}

Fig. \ref{fig_comb} shows that the comb form of PRS also affects the ranging accuracy. It is worth mentioning that the RMSEs of range become unchangeable when SNR reaches a certain threshold in Fig. \ref{fig_Compare} and Fig. \ref{fig_comb}, it is comprehensible according to the estimated range given in (13). The (13) shows that since the peak index $in{d_{{S_g},m}}$ can only be taken as an integer, there will be a minimum error between the estimated range and the real range. With the increase of SNR, The RMSEs of range gradually converges to the minimum error and will not change, which is the best accuracy that can be achieved under high SNR caused by the resolution of algorithm. With the increase of SNR, the RMSE of Comb 2 converges faster. In the case of low SNR, the ranging accuracy of Comb 2 is better than that of Comb 4. The reason is that Comb 4 inserts fewer subcarriers so that it has a smaller peak value, which is more susceptible to noise. Comb 2 performs better in the links with low and medium SNRs, but its performance advantage gradually decreases with the improvement of SNR. In the scenarios with small delay spread, especially in the area of high SNR, Comb 4 with lower overhead has higher throughput performance.

Fig. \ref{fig_Range} shows that the RMSE of the range estimation and the root CRLB of the range estimation decrease with the increase of SNR, and the RMSE of the range estimation gradually approaches the root CRLB. When introducing the fractional factor $m_a$, the curve converges faster and approaches the CRLB. Fig. \ref{fig_ma} shows the impact of different values of $m_a$ on the improvement of RMSE. Considering the tradeoff between accuracy requirements and complexity, the appropriate value of $m_a$ is greater than 10.
\begin{figure}[!t]
\centering
\includegraphics[width=0.45\textwidth]{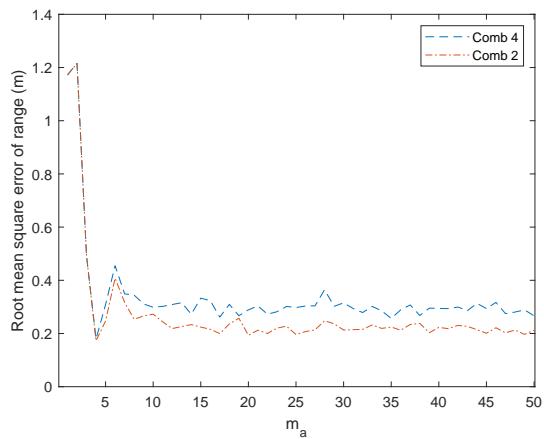}
\caption{RMSEs of range estimations vs $m_a$} \label{fig_ma}
\end{figure}
\subsection{Velocity Estimation}
Supposing that PRS is sent continuously, and the number of symbols is $M=128$. For one measurement result as shown in Fig. \ref{fig_velocitypeak}, the theoretical value of velocity resolution $\Delta v{\rm{ = }}\frac{c}{{2M{T_s}{f_c}}}{\rm{ = 5}}{\rm{.47\;m/s}}$ and the value of $in{d_{{S_g},k}}$ is 3. Thus, the estimated velocity is
\begin{figure}[!t]
\centering
\includegraphics[width=0.45\textwidth]{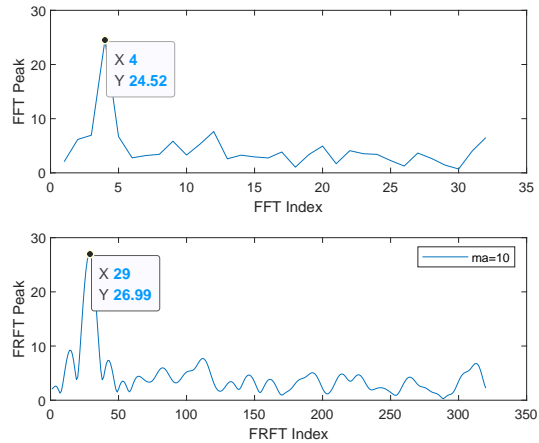}
\caption{Simulated FFT correlation peak when the target's velocity is 15 m/s} \label{fig_velocitypeak}
\end{figure}
\begin{equation}
\hat v = \frac{{in{d_{{S_g},k}}c}}{{2M{T_s}{f_c}}}{\rm{ = 16}}{\rm{.42\;m/s}}.
\end{equation}

Taking the fractional factor ${m_a} = 10$, the peak index $ind_{{S_g},k}^{{m_a}}$ for FFT after introducing ${m_a}$ is 28, the estimated velocity is
\begin{equation}
{\hat v_{{m_a}}} = \frac{{ind_{{S_g},k}^{{m_a}}c}}{{2{m_a} \times M{T_s}{f_c}}}{\rm{ = 15}}{\rm{.33\;m/s}}.
\end{equation}

After 1000 times Monte Carlo simulations, the RMSE of velocity estimation is 2.24 m/s. When introducing the fractional factor, the value of RMSE is reduced to 0.69 m/s. It is revealed that the introduction of fractional factor improves the accuracy of velocity estimation, which reaches the cm/s level.
\begin{figure}[!t]
\centering
\includegraphics[width=0.45\textwidth]{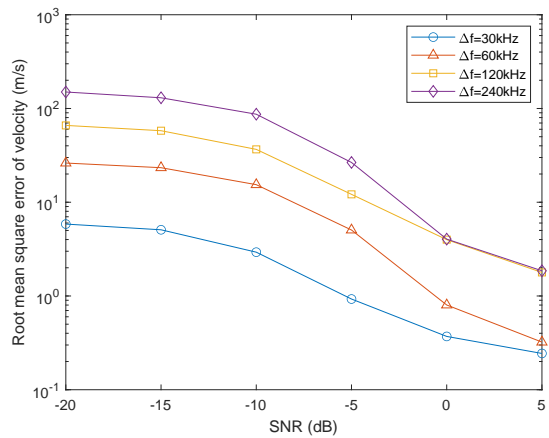}
\caption{Comparison of RMSEs of velocity estimations with different subcarrier intervals} \label{fig_Comparison}
\end{figure}

Fig. \ref{fig_Comparison} shows that the RMSE of velocity estimation decreases with the increase of SNR, and the subcarrier spacing also affects the accuracy of velocity estimation. The accuracy of velocity estimation is increasing with the decrease of the subcarrier spacing. The accuracy of velocity estimation is increasing with the decrease of $\Delta f$. Although the resolution at 120 kHz is 2 times higher than that at 240 kHz, the resolution is still very poor, and the improvement of accuracy is not obvious.
\begin{figure}[!t]
\centering
\includegraphics[width=0.45\textwidth]{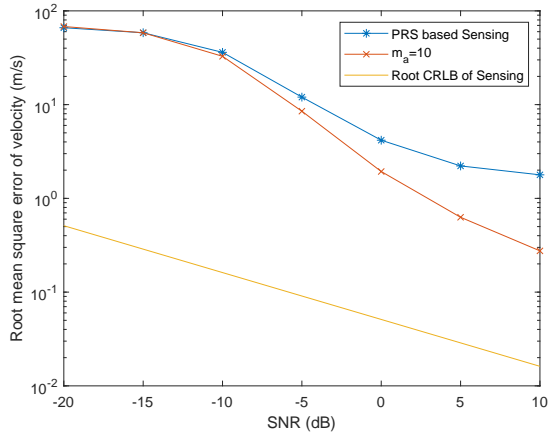}
\caption{RMSEs of velocity estimations vs SNR} \label{fig_Velocity}
\end{figure}

RMSEs of velocity estimations vs SNR are shown in Fig. \ref{fig_Velocity}. The RMSE and the root CRLB of the velocity estimation decrease as the SNR increases. With the introduction of the fractional factor, the RMSE curve of the velocity estimation converges faster.
\subsection{Velocity Measurement using Multiple Frames}
Taking $\Delta f = 60 {\rm{kHz}}$, then a 5G NR frame has 560 symbols. Assuming that the number of symbols occupied by PRS per frame $S_i$ is 128, and the rest symbols are filled with random data modulated by QPSK.

Fig. \ref{fig_overhead} shows the relationship among velocity resolution $\Delta {v_{{\rm{multi}}}}$, overhead in $i$-th frame ${\eta _i}$ and sensing refresh time $\rho$. The velocity resolution $\Delta {v_{{\rm{multi}}}}$ can be improved by increasing the PRS symbol overhead in a single frame ${\eta _i}$ or increasing the sensing refresh time $\rho$. To ensure high velocity resolution, ${\eta _i}$ is reduced by sacrificing the sensing refresh time $\rho$, namely, using more frames for velocity measurement.

Fig. \ref{fig_multi} shows the performance of velocity measurement using multiple frames with different SNRs. When the number of frames is the same, the RMSE of velocity measurement decreases with the increase of SNR, improving the anti-noise performance. With the increase of the number of frames, the accuracy of velocity estimation is improved. Since $S_i$ is 128, the PRS symbol overhead ${\eta _i}$ in a single frame is 22.9\%. In order to ensure the accuracy of velocity measurement, assuming that three frames are used for velocity measurement, the sensing refresh time $\rho$ is 0.03 s so that the vehicle travels 0.45 m during velocity estimation with the velocity of 54 km/h, which is tolerable.
\begin{figure}[!t]
\centering
\includegraphics[width=0.45\textwidth]{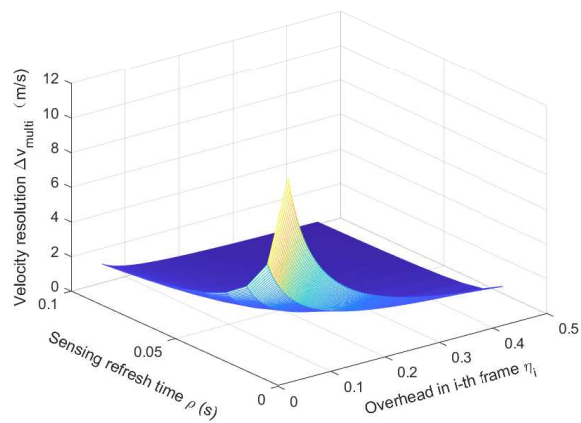}
\caption{Relationship among velocity resolution $\Delta {v_{multi}}$, overhead in $i$-th frame ${\eta _i}$ and sensing refresh time $\rho$} \label{fig_overhead}
\end{figure}
\begin{figure}[!t]
\centering
\includegraphics[width=0.45\textwidth]{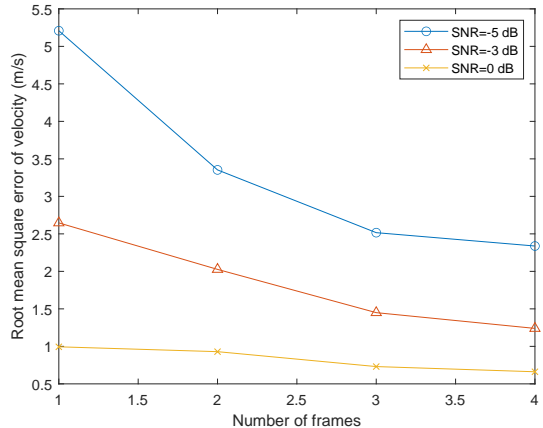}
\caption{RMSEs of velocity estimations vs number of frames} \label{fig_multi}
\end{figure}

\subsection{Inspiration of Frame Structure Design for 5G-A and 6G}
The analysis of PRS for radar sensing shows that it is feasible for radar sensing and positioning. We provide some suggestions for the frame structure design of 5G-A and 6G in the future.
\begin{figure*}[!t]
\centering
\includegraphics[width=1.1\textwidth]{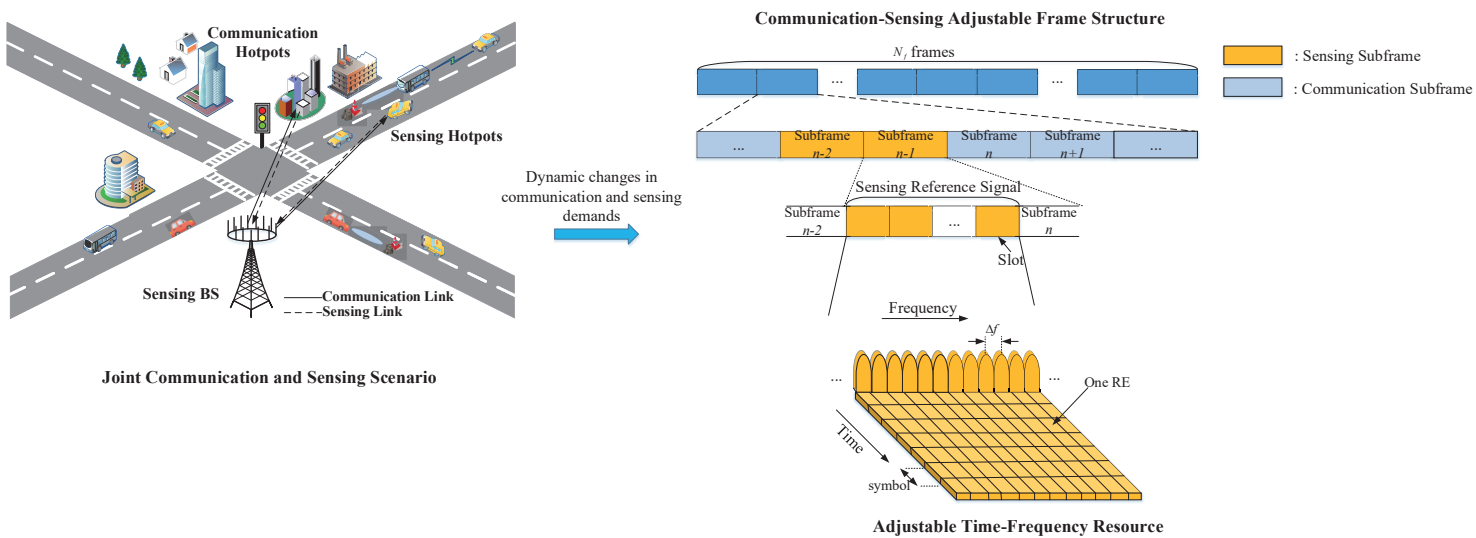}
\caption{Frame structure design for 5G-A and 6G}\label{fig_RS}
\end{figure*}

As shown in Fig. \ref{fig_RS}, we divide the frame structure into sensing subframe and communication subframe. The sensing subframe mainly realizes the function of radar sensing, and the communication subframe is mainly used for data transmission. It can also be used for radar sensing when the demand for sensing performance is high. Braun $et\;al.$ \cite{TR:CR34} used the data of one OFDM frame for single target sensing. For massive downlink data transmission in large bandwidth, data signal can be used for sensing instead of sensing reference signal while serving communication users. However, when a single user transmits with partial bandwidth, the amount of data is small, so that the performance of sensing using data signal is poor \cite{TR:CR29}. Wild $et\;al.$ \cite{TR:CR18} pointed out that although data symbols can be used for sensing, additional radar symbols are still required because there may be no data to transmit at a given time and direction where sensing needs to be performed. Barneto $et\;al.$ \cite{TR:CR30} indicated that the communication transmitter is not always full-load, and there are usually unused subcarriers, which can be used by the radar subcarriers to improve the performance of radar sensing. Therefore, it is necessary to set the sensing reference signal specially used for radar sensing in the sensing subframe. Moreover, the power of pilot signal is high which has good anti-noise performance. In the scenario with high demand for sensing performance, the sensing reference signal and data signal are combined to realize long-distance sensing \cite{TR:CR31}.

In order to realize adaptive communication and sensing performance, the design of sensing reference signal will support waveform reconfiguration to meet the requirements of sensing performance in different scenarios. In the scenario with low sensing performance requirements, the sensing reference signal can use the reference signal of the existing communication system, such as PRS, or redesign the sensing subframe by adjusting the time-frequency resources. Comb pattern is used to orthogonalize the frequency-domain allocation of sensing reference signals of users, which can not only meet the sensing performance required by diversified scenarios, but also further suppress the interference between BSs.

In the sensing subframe, the time-frequency resource allocation of the sensing reference signal is mainly determined by the sensing performance requirements. (32) and (33) show that there is a performance trade-off between range and velocity estimation. Fig. \ref{fig_CRLB} shows the relationship between the CRLB of range and velocity estimation under different SNRs when the subcarrier spacing configuration $\mu$ is configured as 0, 1, 2, 3, 4. As the SNR increases, the CRLB of range and velocity estimation decreases, which indicates that a better lower bound of accuracy can be reached. The parameters such as $\Delta f$, $N$, $M$ and ${f_c}$ can be reasonably configured to satisfy the accuracy requirements for range and velocity estimation. In addition, the time-frequency resource mapping of the sensing reference signal can be flexibly adjusted to cope with different sensing scenarios. Take PRS for example, the different comb structures of PRS can prevent the interference of the PRS signals sent by multiple BSs.
\begin{figure}[!t]
\centering
\includegraphics[width=0.45\textwidth]{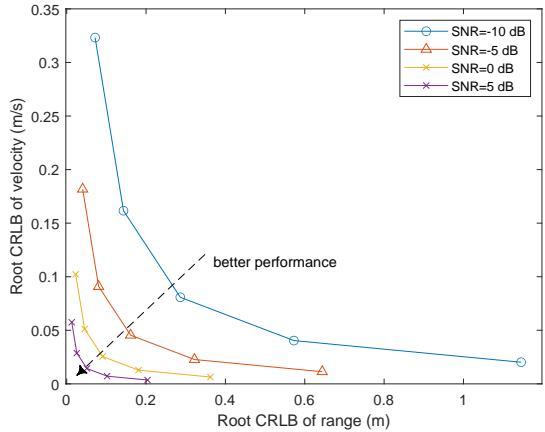}
\caption{The relationship between CRLB of range and velocity estimation} \label{fig_CRLB}
\end{figure}

In addition, in order to enhance the sensing performance, chirp \cite{TR:CR32,TR:CR33} and other frequency modulation schemes can be combined with OFDM to redesign the sensing reference signal. By increasing the time bandwidth product of the signal, the signal has higher resolution and accuracy.
\section{Conclusion}
In this paper, we demonstrate the 5G PRS, as a pilot signal in 5G mobile communication system, has the feasibility and superiority in radar sensing compared with other pilot signals. We investigate the method of using 5G PRS for radar sensing to realize joint sensing, communication and positioning. Therefore, the PRS can be regarded as a sensing reference signal which requires no additional changes to the 5G signal, so that PRS for sensing can be well compatible with mobile communication system. We analyze the range and velocity estimation performance of PRS for radar sensing and derive the corresponding CRLB. Meanwhile, we also propose an enhanced method to improve the accuracy of velocity estimation using multiple frames. Simulation results show that the RMSE of the range and velocity estimation decrease with the increase of SNR. At the end of the paper, we also give some suggestions for the future 5G-A and 6G frame structure design, hoping to provide some inspirations for future frame structure design work.

\section*{Acknowledgment}
The authors appreciate Future Research Lab at China Mobile Research Institute for the support to this research.

\end{document}